\documentclass[copyright,creativecommons]{eptcs}
\providecommand{\event}{ICE 2012} 
\usepackage{breakurl}             

\usepackage[english]{babel}
\usepackage{amsmath, amssymb, amsthm}
\usepackage[all]{xy}

\newtheorem{lemma}{Lemma}
\newtheorem{corollary}{Corollary}
\newtheorem{theorem}{Theorem}
\theoremstyle{definition}
\newtheorem{definition}{Definition}
\newtheorem{example}{Example}

\title{Interacting via the Heap in the Presence of Recursion}

\author {Jurriaan Rot
         \institute{LIACS, Leiden University\\
         The Netherlands}
         \email{jrot@liacs.nl}
         \and
         Irina M\u{a}riuca As\u{a}voae
		 \institute{Faculty of Computer Science \\
		 Alexandru Ioan Cuza University\\
		 Romania}
		 \email{mariuca.asavoae@info.uaic.ro}
		 \and
         Frank de Boer
         \institute{Centrum Wiskunde en Informatica (CWI)\\
         The Netherlands}
         \institute{LIACS, Leiden University\\
         The Netherlands}
         \email{frb@cwi.nl}
         \and
         Marcello M. Bonsangue
         \institute{LIACS, Leiden University\\
         The Netherlands}
         \email{marcello@liacs.nl}
         \and
         Dorel Lucanu
		\institute{Faculty of Computer Science \\
		Alexandru Ioan Cuza University\\
		Romania}
		\email{dlucanu@info.uaic.ro}
         }

\def\titlerunning{Interacting via the Heap in the Presence of Recursion}
\def\authorrunning{J. Rot, I.M. As\u{a}voae, F.S. de Boer, M.M. Bonsangue, D. Lucanu}

\providecommand{\urlalt}[2]{\href{#1}{#2}}
\providecommand{\doi}[1]{doi:\urlalt{http://dx.doi.org/#1}{#1}}

\begin{document}
\maketitle

\begin{abstract}
Almost all modern imperative programming languages include operations
for dynamically manipulating the heap, for example by allocating and
deallocating objects, and by updating reference fields. In the presence
of recursive procedures and local variables the interactions of a
program with the heap can become rather complex, as  an unbounded number
of objects can be allocated either on the call stack using local variables,
or, anonymously, on the heap using reference fields. As such a static
analysis is, in general, undecidable.

In this paper we study the verification of recursive programs with unbounded
allocation of objects, in a simple imperative language for heap manipulation.
We present an improved semantics for this language, using
an abstraction that is precise. For any program with a bounded
visible heap, meaning that the number of objects reachable from
variables at any point of execution is bounded, this abstraction
is a finitary representation of its behaviour, even though
an unbounded number of objects can appear in the state.
As a consequence, for such programs model checking is decidable. Finally we
introduce a specification language for temporal properties
of the heap, and discuss model checking these properties
against heap-manipulating programs.

\end{abstract}

\section{Introduction}

One of the major problems in model checking recursive programs
which manipulate dynamic linked structures is that the state space
is infinite, since programs may allocate an unbounded number of
objects during execution by updating reference fields (pointers). Indeed
model checking and reachability for such programs are undecidable, in general.
Consequently to allow a restricted form of model checking we need to impose
either some syntactic restrictions on the program \cite{Dam96} or some
suitable bounds on its model.  A natural bound for model checking programs
without necessarily restricting their capability of allocating an unbounded
number of objects is to impose constraints on the size of the \emph{visible}
heap~\cite{BFQ07}. The visible heap
consists of those objects which are reachable from the variables
in the scope of the currently executed procedure. Such a bound still allows
for storage of an unbounded number of objects onto the call-stack,
using local variables.

In this paper we introduce a method for model checking
sequential imperative programs with pointers and recursive procedure calls.
In order to allow implementation of model checking of unbounded object allocation in the context
of a bounded visible heap, we introduce a new mechanism for the generation of fresh
object identities which allows for the \emph{reuse}
of object identities and which includes a \emph{renaming} scheme to resolve possible
resulting name clashes.
We introduce a formal operational semantics based on this
mechanism for an abstract programming language, called Shylock. Subsequently
we introduce a logic for reasoning about properties of the heap, where
we use atomic propositions defined as regular expressions in what
is basically a Kleene algebra with tests \cite{Kozen:1997:KAT:256167.256195}.
Namely, the global and local variables of a program  are used as \emph{nominals},
whereas the pointers (reference fields) constitute the set of basic actions.

Our renaming mechanism allows a different kind of reuse of object identities than
usual garbage collection techniques. A garbage collector typically reuses object
identities from the heap and considers objects on the call stack as still in use.
In contrast our technique is more tailored towards model checking, and as such,
we need to reuse as much object identities as possible to guarantee a representation
of program behaviour in terms of a finite pushdown system with a finite
stack alphabet. In fact, our mechanism allows to reuse objects allocated in the call stack,
that may become active when procedures return.

\paragraph{Structure of the paper} In the next paragraph we briefly discuss related work. 
We introduce Shylock and its formal semantics in Section~\ref{defShy}.
In Section~\ref{sec-impr} the abstraction of this semantics is introduced, together
with a proof of its correctness. Then in Section~\ref{sec-mc} we define a logic
for temporal properties of heaps, and finally in Section~\ref{sec-concl} we conclude.

\paragraph{Related work}
We introduce a novel technique for resolving name clashes in the context of reuse
of object identities. It is based on the concept of \emph{cut points} as introduced in
\cite{DBLP:conf/popl/RinetzkyBRSW05} to support static analysis via abstract interpretation
techniques. Cut points are objects in the heap that are referred to from both local
and global variables, and as such are subject to modifications during a procedure call.
Recording cut points in extra logical variables allows for a precise abstract execution
of the program, which in  case of a bound on the visible heap can be represented by a finitary
structure, namely that of a \emph{finite pushdown system}. 

In~\cite{BFQ07} a language is studied with the same features as our Shylock  programs extended with a bounded form of concurrency. Because concurrency is an orthogonal dimension to the
vertical growing of the number of objects due to recursion and the horizontal growing
due to the anonymous field update, we have decided not to incorporate it in our
Shylock language. In fact, bounded concurrency could easily be handled with a
technique similar to the one used in~\cite{BFQ07}. The novelty of our work is not
in the decidability result, which is indeed similar to that obtained in~\cite{BFQ07},
but in the technique we used to obtain it. While~\cite{BFQ07} uses finite graphs and
graph isomorphisms to represent heaps and avoid name clashes, respectively, our approach
is purely symbolic, and, therefore, directly usable for model checking temporal
properties of heaps. We discuss the relationship with~\cite{BFQ07} in more detail in the final paragraph of Section \ref{sec-mc}.

Currently there are several model checkers for object oriented languages.
Java Path Finder~\cite{HP00} is basically a Java Virtual Machine that
executes a Java program not just once but in all possible ways, using
backtracking and restoring the state during the state-space exploration.
Even if Java Path Finder is capable of checking every Java program, the
number of states stored during the exploration is a limit on what can be
effectively checked. As with JCAT~\cite{DIS99}, Java source code can be
translated into Promela, the input language of SPIN. Since
Promela does not support dynamic data structures, fixed-size
heaps and stacks have to be allocated.

Bandera~\cite{CDH00} is an integrated collection of tools for
model checking concurrent Java software using state-of-the art
abstraction, partial order reductions and slicing techniques to reduce the
state space.
It compiles Java source code into a reduced program model expressed in the
input language of other existing verification tools. For example, it can
be combined with the SAL (Symbolic Analysis Laboratory) model
checker~\cite{PSS00} that uses unbounded arrays whose sizes vary dynamically
to store objects. In order to explore all reachable states model checking
is restricted to Java programs with a bounded (but not fixed a priori)
number of \emph{created} objects.

TOPICS \cite{LS10,FLS09} is a tool which aims to find certain types of bugs in non-recursive C programs which manipulate a restricted type of heaps containing only single-linked lists. The faults detected  by TOPICS are of several types, namely: memory leak, segmentation faults, array out of bounds errors, and usage of undefined objects in tests. The method used for bug-detection is based on reachability analysis on counter machines. The transformation of a C program into a counter machine goes through the intermediate representation of pointer machines which abstracts away the contents of the cells of the linked lists.

The problem of reusing object identities has already been faced when
defining semantic models for the pi-calculus, most notably in
history-dependent automata~\cite{MP98}, a model based on the theory of named
sets capable of finite-state verification of processes that can allocate fresh
resources~\cite{CM10}. Model checking of a possibly unbounded number of objects
with pointers but for a language with a restricted form of recursion (tail recursion)
and no block structure has been studied using high level allocation B\"{u}chi
automata~\cite{DKR04} that allow for a finite state symbolic semantics very
similar to ours. Full recursion, but with a fixed-size number of objects is
instead considered in jMoped~\cite{ES01}, using a pushdown structure to
generate an infinite state system.

The techniques described in this paper aim at verifying programs by model checking. This is
fundamentally different from other tools and techniques for verifying programs manipulating
the heap by deductive verification methods, such as separation logic~\cite{Reynolds2002}. Automated methods
for proving annotated programs is a very active area of research (see e.g.~\cite{Madhusudan11,Berdine04,BBR12}). For a more
detailed discussion we refer to~\cite{BBR12}.

\paragraph{Acknowledgments} We would like to thank the anonymous referees for their 
extensive comments and suggestions that greatly improved the presentation of our work.
The research of Jurriaan Rot has been supported by the Dutch NWO project CoRE.

\section{Shylock: a language to manipulate the heap}\label{defShy}

In this section we introduce \emph{Shylock}, a simple imperative programming
language that allows us to focus on dynamic pointer structures in the
context of recursive procedures with local variables. Programs consist of a
set of recursive procedures that can create new objects, and store them into
their local or global variables. Besides being dynamically allocated, objects
can be referenced by other objects via object fields, and exist as long as they are
reachable in the heap from some other object or from a variable. To simplify the presentation
objects are the only data structure of Shylock.

We assume an infinite set $V$ of \emph{variables} ranged over by $x,y$, including a finite
collection $G$ of \emph{global program variables} $\{g_1, g_2, \ldots, g_n\}$, and a
disjoint finite set $L$ of \emph{local program variables} $\{l_1, l_2, \ldots, l_m \}$.
We denote by $C$ the infinite set $V \setminus (G \cup L)$  of \emph{cut point variables} ranged over
by $c_1, c_2, \ldots$. Further, we assume a distinguished element $nil \in G$, used
as a constant to refer to the undefined object. A Shylock program acts only on global and local
variables, cut point variables will be used later in the abstract semantics as a kind
of ``freeze'' variables for storing relevant points of the heap during procedure call.
For simplicity we assume that all objects have the same set of \emph{fields}
$F = \{f_1, \ldots, f_k\}$. We denote by $\bar{g}, \bar{l}, \bar{c}, \bar{f}$ the
sequences of all global, local, cut point variables and fields, respectively.

For $P$ a finite set of \emph{procedure names} $\{p_0, \ldots, p_l\}$, a program is a set
of \emph{procedure declarations} of the form $p_i~::~B_i$, where $B_i$, denoting the \emph{body}
of the procedure $p_i$, is a statement defined by the following grammar:
\begin{align*}
B \; {:}{:}{=} & \; x.f ~{:}{=}~ y \;|\; x ~{:}{=}~ y.f \;|\; x ~{:}{=}~ \text{new} \;|
           \; [x=y] B \;|\; [x\neq y] B \;|\; B+B \;|\; B;~B \;|\; p \,
\end{align*}
Here $x$ and $y$ are (local or global) program variables ranging over $G \cup L$,
$f$ is a field in $F$ and $p$ is a procedure name in $P$. We assume a distinguished
$p_0 \in P$, called the \emph{initial} procedure of a program.

The \emph{assignment} statements $x.f ~{:}{=}~ y$ or $x ~{:}{=}~ y.f$ assign the
identity of the object referenced at right hand side of the assignment
to the field or variable, respectively, at the left hand side. The statement
$x~{:}{=}~ \text{new}$ \emph{creates} a new object that will be referenced by the
program variable  $x$. All fields of $x$ will reference the undefined object $nil$.
We  restrict to programs in which the variable $nil$ does not appear in the
left-hand side of an assignment or object creation, i.e., $nil$ is a constant.
\emph{Conditional statements} $[x=y]B$ and $[x \neq y] B$, \emph{nondeterministic choice}
$B_1 + B_2$, and \emph{sequential composition} $B_1;B_2$, have the standard interpretation.
A \emph{procedure call} $p$ means that the body $B$ associated with $p$ is executed next on
the same global state but on a fresh local state. After the procedure body terminates,
its local state is destroyed forever and the previous local state (from which the procedure has
been called) is restored. Changes to the global state, however, remain.

Notice that variable assignment $x:=y.f$ and field update $x.f:=y$ suffice, as more general expressions and updates
can be encoded. For example, a statement $x := y.f_{i_1} \ldots f_{i_k}$ is
encoded as $x:=y.f_{i_1};x:=x.f_{i_2}; x:=x.f_{i_3}; \ldots ; x:=x.f_{i_k}$. A basic variable assignment
$x := y$ can be encoded as $z.f := y ; x := z.f$.
More general boolean expressions in conditional statements can be obtained by using sequential composition and
nondeterministic choice. In fact $(b_1 \wedge b_2) B$ can be written as $(b_1) b_2 B$,
whereas  $(b_1 \vee b_2) B$ as $(b_1 B) + (b_2 B)$. Negation of a boolean expression
$b$ can be obtained by transforming $b$ into an equivalent boolean expression in conjunctive disjunctive normal form,
for which negation of the simple expression $[x=y]$ and $[x \neq y]$ is defined as expected.
Ordinary while, skip, and if-then-else statements can be expressed easily in the language, using recursive procedures,
conditional statements and nondeterministic choice. For the sake of simplicity, we allow creation and assignment of a
single object identity only; generalizations to simultaneous assignments and object creation can be added in a
straightforward manner.
The language does not directly support parameter passing. However, it is worthwhile to note that we can model
procedures with call-by-value parameters by means of global variables. Let $p(v_1, \ldots, v_n)$ be a procedure with
formal parameters $v_1, \ldots, v_n$. We see the formal parameters as local variables and introduce for each parameter
$v_i$ a corresponding global variable $g_i$ (which does not appear in the given program). Every procedure call
$p(x_1,\ldots, x_n)$ can then be encoded by the statement $g_{1}:=x_1; \ldots; g_{n}:=x_n; p$  whereas the body
$B$ of $p(v_1, \ldots, v_n)$ can be encoded by $v_1 := g_{1}; \ldots; v_n := g_{n}; B$. A similar approach
can be taken to model procedures with return values. Finally, method calls $x.m(x_1,\ldots,x_n)$ can be modeled
by introducing the called object $x$ as an additional parameter of the procedure $m$.

\begin{example}\label{ex-file}
We consider a simple example program which opens a file, passes it to some procedure which returns again
a file, and finally tries to close this returned file\footnote{This idea was suggested to us by Dilian Gurov.}. It
consists of the procedures \emph{main}, \emph{q}, \emph{open} and \emph{close}, and the sets of global and local
variables are $G=\{nil\}$ and $L=\{x,y\}$ respectively. The procedures \emph{main}, \emph{open} and \emph{close} are defined as follows:
\begin{align*}
main & :: \text{open}(x);~y := q(x);~close(y) \\
open & :: x := new \\
close(z) & :: [z \neq nil] z := nil
\end{align*}
The definition of $q$ is left open. Recall from the above discussion that while parameter passing and return values
are not directly in the syntax of the language, they can easily be encoded.
The intuition behind this program is as follows. We start by executing \emph{main}.
Then first the program opens a file, modeled by an object creation, and then it passes the reference $x$ to this file
on to a procedure $q$. This procedure $q$ then performs some calculations and passes back a file reference $y$. Finally
we try to close the file referenced by $y$. Closing a file is modeled by first checking if the given reference
is not \emph{nil}, and then simply setting the reference to $nil$. If we pass \emph{close} a reference to \emph{nil}, then
the program crashes. \qed
\end{example}


In order to describe the formal semantics of Shylock programs, we first formalize some relevant
notions related to the heap. To represent object identities we use the set $\mathbb{N}_\bot = \mathbb{N} \cup \{\bot\}$ of natural numbers
extended with an element $\bot$ and ranged over by $n,m$.
Let $s: V \rightarrow \mathbb{N}_\bot$ be a variable assignment and
$h: F \rightarrow (\mathbb{N}_\bot \rightarrow \mathbb{N}_\bot)$ be a field assignment
such that for all $f$, $h(f)(\bot)=\bot$ and the set of objects for which $h(f)(n) \neq \bot$ is finite.
A \emph{heap} $H$ is a pair $\langle s,h \rangle$ of a variable- and a field assignment. We write $H(x)$ for $s(x)$, and
$H(f)$ for $h(f)$.
For a subset of variables $Var \subseteq V$ we denote with $\mathcal{R}_{H}(Var)$ the set of objects reachable from
objects labeled by these variables in $H$ via any of the (functional) transition relations $H(f)$, for $f \in F$. Formally
it is defined as the least fixpoint of the equation
$$\mathcal{R}_H(Var) = \{H(x) \mid x \in Var\} \cup \{H(f)(n) \mid f \in F, n \in \mathcal{R}_H(Var)\}$$
If $Var=V$ we denote the set of reachable objects of $H$ by $\mathcal{R}_{H}$. Further we define
the ``purely local'' part of a heap $H$ as
$\mathcal{R}^\diamond_H= \mathcal{R}_{H}(L \cup C) \setminus \mathcal{R}_H(G)$. Intuitively
$\mathcal{R}^\diamond_H$ contains all objects which are reachable from a local (or cut point) variable,
but not from a global variable.


We denote variable update by $H[x~{:}{=}~n]$, global field update by
$H[f ~{:}{=}~ \varphi]$ where $\varphi:\mathbb{N}_\bot \rightarrow\mathbb{N}_\bot$ is a function such that
$\varphi(\bot)=\bot$, and local field update by $H[f ~{:}{=}~ \varphi[n ~{:}{=}~ m]]$. We use
the standard notation and definition of simultaneous assignments and updates.
A \emph{renaming} $\rho$ of a subset $N \subseteq \mathbb{N}_\bot$ is an injective function
in $\mathbb{N}_\bot \rightarrow \mathbb{N}_\bot$ such that $\rho(n)=n$ for all $n \not \in N$,
and otherwise $\rho(n) \in N$. Clearly it has an inverse, denoted by $\rho^{-1}$. Given
a renaming $\rho$ we define its application on a heap $H$ as $\rho(H)(x) = \rho(H(x))$ and $\rho(H)(f)(n) = \rho(H(f)(\rho^{-1}(n)))$.

A \emph{configuration} is a tuple $\langle H, \Gamma \rangle$ where $H$ is the current heap and
$\Gamma$ is a stack of statements and heaps. The head of a stack is separated from the tail by means of
the right-associative operator $\bullet$, while the empty stack is represented by $\epsilon$.
The current statement to be executed is on the top of the stack. When there are no statements but an heap
on the top of the stack, then a procedure returns, and the state on the stack has to be restored as
current state. A \emph{computation} is a (possibly infinite) sequence
$C_0 \longrightarrow C_1 \longrightarrow \ldots$ of transitions, where $\longrightarrow$
is a relation between configurations which we will now define by cases on the top of the stack.
To this end let $\Gamma$ be a stack of statements and heaps.
Assignments to a variable or to a field update the current heap structure as expected:
\begin{equation*} \label{ruleassign_to_var}
 \langle H, x ~{:}{=}~ y.f \bullet \Gamma \rangle \longrightarrow
 \langle H[x~{:}{=}~H(f)(H(y))], \Gamma \rangle
\end{equation*}
\begin{equation*} \label{ruleassign_to_field}
 \langle H, x.f ~{:}{=}~ y \bullet \Gamma \rangle \longrightarrow
 \langle H[f ~{:}{=}~ H(f)[H(x)~{:}{=}~H(y)]], \Gamma \rangle
\end{equation*}
To model object creation we assume a distinguished global ``system'' variable $oc$ which is
used as a counter, and does not appear in a program. We implicitly assume that $H(oc)\not=\bot$,
for every heap $H$. The semantics of the operator ``new'' is:
\begin{equation*} \label{rulecreate}
 \langle H, x ~{:}{=}~ \text{new} \bullet \Gamma \rangle  \longrightarrow
 \langle H[x,oc ~{:}{=}~ H(oc), H(oc)+1][\bar{f} ~{:}{=}~ \bar{\varphi}], \Gamma \rangle
\end{equation*}
where $\bar{\varphi}$ is the sequence such that $\varphi_i = H(f_i)[H(x) ~{:}{=}~ \bot]$.
Conditional statements are executing depending on the evaluation of the condition:
\begin{equation*}\label{rulecond}
 \frac{H(x)=H(y)}
      {\langle H, [x=y] B \bullet \Gamma \rangle \longrightarrow \langle H, B \bullet \Gamma \rangle}
     ~~~~~~~~~~~~~~~~~~
 \frac{H(x) \neq H(y)}
      {\langle H, [x \neq y] B \bullet \Gamma \rangle \longrightarrow \langle H,B \bullet \Gamma \rangle}
\end{equation*}
Sequential composition adds on top of the stack the next statements to be executed, while
nondeterministic choice selects just one of the two statements.
\begin{equation*}\label{rulecontrol}
 \langle H, B_1; B_2 \bullet \Gamma \rangle \longrightarrow
 \langle H, B_1 \bullet B_2 \bullet \Gamma \rangle
 ~~~~~~~~~~~~~~~~~~
 \langle H, B_1 + B_2 \bullet \Gamma \rangle \longrightarrow \langle H, B_i \bullet\Gamma \rangle ~~~~~~~ i \in \{0,1\}
\end{equation*}
Finally, procedure call and return are modeled as follows:
\begin{equation*} \label{ruleproccall}
 \langle H, p \bullet \Gamma \rangle \longrightarrow
 \langle H[\bar{l} ~{:}{=}~ \bar{\bot}], B \bullet H \bullet \Gamma \rangle
 ~~~~~~~~~~~~~~~~~~
 \langle H, H' \bullet \Gamma \rangle \longrightarrow
 \langle H[\bar{l} ~{:}{=}~ H'(\bar{l})], \Gamma \rangle
\end{equation*}
where $H'(\bar{l})$ denotes the pointwise application of $H'$ to the local variables $\bar{l}$, and $B$
is the body of the procedure $p$. Recall that for technical convenience there is a single sequence of local variables $\bar{l}$ shared
by the procedures.

This section is concluded with the notion of \emph{properness}, which is a formalization of some operational properties of configurations
appearing during the execution of programs.
Suppose $H$ and $H'$ are heaps which appear in the stack, meaning they are pending heaps from
a procedure call, and $H$ appears higher up in the stack than $H'$ (so $H'$ was put on the stack before $H$). Then (1) the structure
of the purely local part of $H'$ is preserved in $H$, since it could not have been accessed in between. Moreover (2) if
some object is reachable in both $H$ and $H'$, then it must be reachable from a global variable in $H'$. Both conditions
intuitively capture the property that reachable objects remain reachable in a recursive call only if they
were already reachable from global variables.
\begin{definition}
The set of \emph{proper} stacks is defined inductively as follows:
 \begin{itemize}
  \item the empty stack is proper
  \item if $\Gamma$ is proper then $B \bullet \Gamma$ is proper for statements $B$
  \item if $\Gamma$ is proper and $H$ is a heap such that for every $H'$ occurring in $\Gamma$ the following holds:
      \begin{itemize}
	\item for all $f \in F$, $n \in \mathcal{R}^\diamond_{H'}$: $H(f)(n) = H'(f)(n)$
	\item $\mathcal{R}_{H'} \cap \mathcal{R}_{H} \subseteq \mathcal{R}_{H'}(G)$
      \end{itemize}
      then $H \bullet \Gamma$ is proper.
 \end{itemize}
 A configuration $\langle H, \Gamma \rangle$ is proper if $H \bullet \Gamma$ is a proper stack.
\end{definition}

For example every configuration $\langle H,p_0 \rangle$ is proper. Further, all transition steps preserve
proper configurations. Thus every configuration in a computation starting from a proper one is proper.

\section{Improving the semantics}\label{sec-impr}

The semantics introduced may generate, for a given program, a transition system with infinitely many configurations.
This is not only because of the unbounded stack size, but, more problematically, also because each time a new object
is allocated a new natural number is used. Thus, the number of heaps needed is also unbounded.
Consider for example the Shylock program
consisting of a single procedure $p$ with as body the statement
\[
x ~{:}{=}~ \text{new}; \; p
\]
where $x$ is a local variable.
Each time the statement $x~{:}{=}~$new is executed, a new natural number
is assigned to $x$.
This has the unfortunate consequence that infinitely many heaps are needed, and thus
the usual model checking techniques for recursive
systems~\cite{ES01,Sch02b}
cannot be guaranteed to terminate.
In this section we introduce an abstract semantics for Shylock, based on \emph{reuse} of natural numbers for objects.
Identities of objects that
are not in use in the \emph{current} heap will be reused instead of using a new
identity each time a new object is allocated. More concretely, when creating
a new object we will choose the \emph{minimal} unused identity available, as formally
expressed by the following rule:
\begin{equation*}
 \langle H, x ~{:}{=}~ \text{new} \bullet \Gamma \rangle  \longrightarrow \langle  H[x~{:}{=}~n][\bar{f} ~{:}{=}~ \bar{\varphi}], \Gamma \rangle
 \end{equation*}
where $n = \min(\mathbb{N} \setminus\mathcal{R}_H)$, and, for each $i$, $\varphi_i = H(f_i)[n ~{:}{=}~ \bot]$.
The intuition here is that numbers are no longer concrete object identities, but instead they represent equivalence
classes. However, this adapted rule may introduce name clashes with objects in the local state
pending on the stack. We illustrate this problem and our solution with an example. Consider the following
heap:
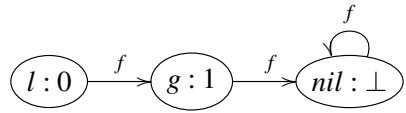
\begin{equation*}
\xymatrix{
*++[o][F-]{l:0} \ar[r]^f & *++[o][F-]{g:1} \ar[r]^f & *++[o][F-]{nil:\bot} \ar@(ru,ul)[]_f
}
\end{equation*}
Here $x:n$ represents the identity $n$ to which the variable $x$ refers (so technically
the figure represents a heap $H$ for which $H(l)=0$, $H(f)(0) = 1$, etc.). Further
$l$ is a local variable, $g$ is a global variable, and $f$ is the single field.
Let us consider first the execution of a call to a procedure $p~::~g ~{:}{=}~ \text{new}$.
Starting from the above heap, on the call a copy is placed onto the stack and the
local variable $l$ is initialized to $\bot$, so the procedure $p$ is executed on the following heap:
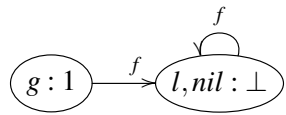
\begin{equation*}
 \xymatrix{
  *++[o][F-]{g:1} \ar[r]^f & *++[o][F-]{l,nil:\bot} \ar@(ru,ul)[]_f
}
\end{equation*}
When executing $g ~{:}{=}~ \text{new}$ we take for $g$ the minimal object identity unreachable
from the current variables, which is $0$. Then, on procedure return, we see that there
is a name clash: both $g$ and $l$ point to the object with identity $0$, while they
should obviously not be identified. A solution is to \emph{rename} the object $n$ to which $g$ points, i.e.,
to make $g$ point to an identity $m$ which is used neither by the current nor the caller's stored heap, and
updating the fields of this new object according to the fields of $n$. Then we can just take the union of the global
part of the new heap, and the local part of the stored heap. For example we could rename, in the current heap, the
object 1 to 2, which is free, take the union with the (local part of the) stored heap of the caller, and combine
the two heaps as follows:
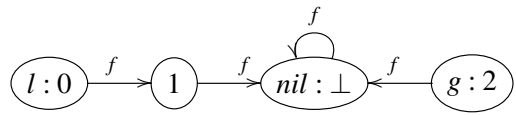
\begin{equation*}
\xymatrix{
*++[o][F-]{l:0} \ar[r]^f & *++[o][F-]{1} \ar[r]^f & *++[o][F-]{nil:\bot} \ar@(ru,ul)[]_f & *++[o][F-]{g:2} \ar[l]_f
}
\end{equation*}
However, consider now the execution of a procedure $p'~::~g ~{:}{=}~ \text{new}; \; g ~{:}{=}~ \text{new}$,
starting from the same heap as before (the first figure above). After executing the first object creation statement in the procedure,
$g$ again points to 0. But then after the second time we execute $g ~{:}{=}~ \text{new}$, $g$ is assigned the minimal
index available, which is 1 at that point. Thus on procedure return, the heap is exactly the same as in the beginning
of the procedure execution. So two object creations are in this case indistinguishable
from no creation at all. On the procedure return, when combining the current heap with the stored heap
it is thus not clear whether $g$ should keep pointing to 1 (when no object creation statements were executed), or if it should be renamed to a new identity
separate from the others (when two object creation statements were executed).


Our solution to this problem is a non-trivial extension of the semantics of procedure call and -return
based on the identification of so-called \emph{cut points}, which
are the object identities at the ``edge'' of the global and the local part of the heap,
representing exactly the point where the local part ``enters'' the global part. On
a procedure call, these cut points are identified, and we assign their values to a set of distinguished
\emph{cut point variables}, in the heap of the \emph{callee}. Then, on procedure return,
the cut point variables ``connect'' the current heap with the stored one, giving us
precisely the information about how to combine the two. Returning to our example, consider the
following heap:
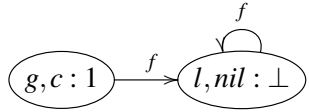
\[
\xymatrix{
*++[o][F-]{g,c:1} \ar[r]^f & *++[o][F-]{l,nil:\bot} \ar@(ru,ul)[]_f &
}
\]
This heap represents the initial heap of the callee, extended with the only cut point
of the caller heap, in the form of the cut point variable $c$. Now if we execute
$g~{:}{=}~ \text{new}; g~{:}{=}~ \text{new}$, the global variable $g$ is not assigned the identity
$1$ since it is already in use. So now on procedure return we can distinguish between the case
that $g$ was newly created (in which case it will have a new identity), and the case that it was not (in which
case it will have the same identity as before).

Formally, for a given heap $H$, the set $CP_H$ of cut points
is defined as follows: $$\mathcal{R}_H (G) \cap (H(L \cup C) \cup F(\mathcal{R}_H^\diamond))$$
where $F(N) = \{H(f)(n) \mid n \in N, f \in F\}$.
Further, $H(V)$ means applying $H$ point-wise to $V$, hence $H(L \cup C) =\{ H(v) \mid v \in L\cup C \}$.
Note that the definition involves the cut point
\emph{variables}; recall from the above discussion that these variables represent the
cut points of the \emph{previous} heap.
Further recall that $\mathcal{R}_H(G)$ is the global part of the heap, while $\mathcal{R}_H^\diamond$ is the ``purely local'' part of
the heap. Intuitively, $F(\mathcal{R}_H^\diamond)$ represents the objects which are pointed to by a field from an object which is purely
local. Further $H(L \cup C) \cap \mathcal{R}_H(G)$ is the set containing objects pointed to \emph{directly} by the local variables,
which are also reachable from global variables. On the other hand, $F(\mathcal{R}_H^\diamond) \cap \mathcal{R}_H(G)$ contains objects adjacent to the
purely local (reachable) nodes (where the node adjacency is provided by the field transitions).

Now the procedure call of the improved semantics is modeled by the following rule.
\[
\langle H, p_i \bullet \Gamma \rangle  \longrightarrow \langle H[\bar{v} ~{:}{=}~ \bar{\bot}][\bar{c} ~{:}{=}~ \bar{n}], B_i \bullet H \bullet \Gamma \rangle  \,,
\]
where $\bar{v}$ is the sequence of local variables and cut point variables $c$ for which $H(c) \neq \bot$. Further
$\bar{c}$ is a sequence of cut point variables of the same length as the sequence $\bar{n}$ of cut points $CP_H$.
Note that, given a heap $H_c$, if on top of the stack  we have a heap $H_l$,  by the way we modeled
the  procedure call, the cut points in the stacked heap $H_l$ correspond exactly to the cut point variables
in the current heap $H_c$.

We proceed to discuss the construction of the return heap, say $H_r$. We first rename all objects of the
purely local part of $H_l$ which conflict with $H_c$, meaning
that they are also reachable from global variables in $H_c$. When this is done, we can just copy all the
local variables directly to $H_c$ and update the fields of the purely local part of $H_l$ in $H_c$.
In order to formalize this process we define the set of \emph{name clashes} $N$ of $H_c$ and $H_l$: $$N = \mathcal{R}^\diamond_{H_l} \cap \mathcal{R}_{H_c}(G)$$
Remember that $\mathcal{R}^\diamond_{H_l}$ only contains the objects which are
reachable from a local variable (or from a previous cut point variable), and are
not reachable from any global variable. Now the return rule is formalized as follows:
\[
\langle  H_c, H_l  \bullet \Gamma\rangle  \longrightarrow \langle  \theta \circ \rho(H_c), \Gamma \rangle
\]
where
\begin{itemize}
 \item $\rho$ is a renaming, monotonic on $N$, such that $\rho(n) \in \mathbb{N} \setminus \mathcal{R}_{H_c[\bar{l} ~{:}{=}~ \bar{\bot}]}$
 if $n \in N$, and $\rho(n)=n$ otherwise,  and $\rho$ is minimal w.r.t pointwise comparison between renaming functions
 \item $\theta$ resets the purely local part:
  $$\theta(H)= H[\bar{l},\bar{c} ~{:}{=}~ H_l(\bar{l}),H_l(\bar{c})][\bar{f}~{:}{=}~\bar{\varphi}]$$
where $\bar{c}$ is the sequence of cut point variables $c$ for which $H_l(c) \neq \bot$, and
$\bar{\varphi}$ is the sequence defined, for all $n \in \mathbb{N}_\bot$, as follows:
$$\varphi_i(n) ~{:}{=}~
\begin{cases}
 H_l(f_i)(n) & \text{ if } n \in \mathcal{R}_{H_l}^\diamond \\
 H_c(f_i)(n) & \text{ otherwise }
\end{cases}
$$
\end{itemize}

\paragraph{Correctness}\label{sec-correctness}

We provide a proof of the equivalence between the concrete semantics of Shylock, and the abstract semantics defined
in the previous section. To this end
we adapt the concrete semantics to take into account the initialization and restoration of the cut point variables,
similar to the abstract semantics.
Note that this does not affect the behaviour of programs, as they are assumed not to contain cut point variables.
First we need the basic notion of heap isomorphism:
\begin{definition}
Two heaps $H$ and $H'$ are \emph{isomorphic}, denoted $H \sim H'$, if there exists a function $\alpha: \mathcal{R}_H \rightarrow \mathcal{R}_{H'}$
such that
\begin{itemize}
 \item $\alpha$ is a bijection.
 \item For each $x \in V$: $\alpha(H(x)) = H'(x)$.
 \item For each $f \in F$ and $n \in \mathcal{R}_H$: $\alpha(H(f)(n)) = H'(f)(\alpha(n))$.
\end{itemize}
\end{definition}
Note that since fields are deterministic, such a function $\alpha$, if it exists, is unique.
In order to proceed, we introduce the important notion
of \emph{cut point identification}. Recall that
on a procedure call, the new heap $H_c$ represents
in its cut point \emph{variables} the cut points of the
caller heap $H_l$. Suppose now we have other heaps $H_c'$ and
$H_l'$ such that $H_c \sim H_c'$ and $H_l \sim H_l'$. Note
that cut points are preserved by isomorphisms. Cut point
\emph{identification} now formalizes the representation of cut points
of $H_l$ and $H_l'$ in $H_c$ and $H_c'$ respectively, using for
each cut point in $H_l$ and $H_l'$ the \emph{same} variable to represent it in $H_c$ and $H_c'$.
\begin{definition}[Cut point identification]
Let $H_c,H_l,H_c',H_l'$ be heaps such that
$H_c \sim_{\alpha_c} H_c'$, $H_l \sim_{\alpha_l} H_l'$. Let
$\{n_1, \ldots, n_k\} = CP_{H_l}$ be the cut points of $H_l$.
We define
$$(H_c,H_l) \bowtie (H_c',H_l')$$
iff there exists a sequence of cut point variables $c_1, \ldots, c_k$ such that for all $i \leq k$:
$$H_c(c_i) = n_i \text{ and }  H_c'(c_i) = \alpha_l(n_i)$$
\end{definition}
From this definition we immediately deduce that the two isomorphisms agree on cut points:
\begin{corollary}\label{corr-cpi}
 If $(H_c,H_l) \bowtie (H_c',H_l')$ then for all $n \in CP_{H_l}$:
$\alpha_l(n)=\alpha_c(n)$.
\end{corollary}
\begin{proof}
 For all $n_i \in CP_{H_l}$ we have $\alpha_l(n_i) = H_c'(c_i) = \alpha_c(H_c(c_i)) = \alpha_c(n_i)$.
\end{proof}
Now we are ready to introduce a strong notion of equivalence, based on heap isomorphism,
which also takes along the main operational properties characterizing
configurations appearing in computations.
\begin{definition}
Given stacks $\Gamma, \Gamma'$ we define $\Gamma \sim \Gamma'$ inductively as follows:
\begin{itemize}
 \item if $\Gamma$ and $\Gamma'$ are both empty then $\Gamma \sim \Gamma'$
 \item if $\Gamma \sim \Gamma'$ then $B \bullet \Gamma \sim B \bullet \Gamma'$
 \item if \begin{itemize}
  \item $\Gamma \sim \Gamma'$, $H \sim H'$
  \item $(H,\nu(\Gamma)) \bowtie (H',\nu(\Gamma'))$
  \item $H \bullet \Gamma$ is proper
 \end{itemize}
  then $H \bullet \Gamma \sim H' \bullet \Gamma'$
\end{itemize}
where $\nu(\Gamma)$ extracts from the stack the top heap.
Now for configurations we define $\langle H, \Gamma \rangle \sim \langle H', \Gamma' \rangle$ iff
$H \bullet \Gamma \sim H' \bullet \Gamma'$.
\end{definition}
The following lemma states how the current heap and the stacked heap are combined on a procedure
return in the concrete semantics. More precisely it expresses that any identity which \emph{becomes}
reachable right after a procedure returns, is in the purely local part of the heap of the caller procedure.
\begin{lemma}\label{lm-return-localpart}
Suppose $\langle H_c, H_l \bullet \Gamma \rangle \sim \langle H_c', H_l' \bullet \Gamma' \rangle$. Let
$H_r = H_c[\bar{l} := H_l(\bar{l})]$. Then for all $n \in \mathcal{R}_{H_r}$:
if $n \not \in \mathcal{R}_{H_c}(G \cup C)$ then $n \in \mathcal{R}_{H_l}^\diamond$.
\end{lemma}

\begin{proof}
Let $H_c, H_l, H_r$ be as above.
Let $n \in \mathcal{R}_{H_r}$ and assume $n \not \in \mathcal{R}_{H_c}(G \cup C)$, so $n$
is reachable in $H_r$ from a local variable.
We prove by induction that any path in $H_r$ reaching such an $n$
is reflected in the same path in $H_l$ which lies entirely in its purely local part $\mathcal{R}_{H_l}^\diamond$.

Suppose first that $n=H_r(l)$ for some local variable $l$. Then $n=H_r(l)=H_c[\bar{l} := H_l(\bar{l})](l) = H_l(l)$.
Now by assumption (that is in the relation $\sim$), $H_c$ has cut point \emph{variables} precisely on the cut \emph{points} of $H_l$. We may then
conclude that $n$ is not reachable from a global variable in $H_l$; otherwise it would, by definition (of cut points) be
on a cut point of $H_l$, and consequently on a cut point variable in $H_c$ which contradicts our assumption on
the reachability of $n$. Thus $n \in \mathcal{R}_{H_l}^\diamond$.


Now let $n=H_r(f_i) \circ \ldots \circ H_r(f_1)(H_r(l))$
such that $n \not \in \mathcal{R}_{H_c}(G \cup C)$,
$n=H_l(f_i) \circ \ldots \circ H_l(f_1)(H_l(l))$ and $n \in \mathcal{R}_{H_l}^\diamond$.
Suppose $H_r(f)(n) \not \in \mathcal{R}_{H_r}(G \cup C)$ for some field $f$.
Since $n \in \mathcal{R}_{H_l}^\diamond$ and $H_c \bullet H_l$ is proper we have
$H_l(f)(n) = H_r(f)(n)$. Now $H_r(f)(n) \neq H_c(c)$
for all cut point variables $c$; otherwise $n$ would be reachable from such an $H_c(c)$ which would be a contradiction with
our assumption. But then by the cut point identification of $H_c$ and $H_l$, $H_l(f)(n)$ is not on
a cut point of $H_l$, which implies that the global state has not been entered yet, i.e.,
$H_l(f)(n) \in \mathcal{R}_{H_l}^\diamond$ as desired.
\end{proof}
We are now ready for the main theorem of this section, stating that the
concrete and the abstract semantics are equivalent.
\begin{theorem}[Bisimulation]\label{thm-bisimulation}
 Let $C_1$ and $C_2$ be configurations such that $C_1 \sim C_2$. Denote with $\rightarrow_c$ and $\rightarrow_a$
 the transition relations corresponding to the concrete and the abstract semantics, respectively. If
  $C_1 \rightarrow_c C_1'$ then there exists a configuration $C_2'$ such that $C_2 \rightarrow_a C_2'$ and $C_1' \sim C_2'$, and vice versa.
\end{theorem}

\begin{proof}
 We only discuss the isomorphism of the resulting heaps on procedure return.
Suppose
$$\langle H_c, H_l \bullet \Gamma \rangle \sim \langle H_c', H_l' \bullet \Gamma' \rangle$$
By definition of the concrete and the abstract semantics from these respective configurations the enabled transitions are
$$\langle H_c, H_l \bullet \Gamma \rangle \rightarrow_c \langle H_r, \Gamma \rangle
~~~~~ \text{ and } ~~~~~
\langle H_c', H_l' \bullet \Gamma' \rangle \rightarrow_a \langle H_r', \Gamma' \rangle$$
where $H_r' = \theta \circ \rho (H_c')$. By definition of $\sim$ there are isomorphisms $H_c \sim_{\alpha_c} H_c'$ and $H_l \sim_{\alpha_l} H_l'$.
We explicitly define an isomorphism $\alpha: \mathcal{R}_{H_r} \rightarrow \mathcal{R}_{H_r'}$ as follows:
$$
\alpha(n) =
\begin{cases}
 \rho \circ \alpha_c (n) & \text{ if } n \in \mathcal{R}_{H_c}(G \cup C) \\
 \alpha_l(n) & \text{ otherwise }
\end{cases}
$$
Note that by Lemma~\ref{lm-return-localpart}, if $n \not \in \mathcal{R}_{H_c}(G \cup C)$ then
$n \in \mathcal{R}_{H_l}^\diamond$ so $\alpha$ is well-defined. To see
that $\alpha$ is an isomorphism, intuitively, note that
$\alpha_c$ is an isomorphism on $\mathcal{R}_{H_c}(G \cup C)$ and $\alpha_l$
is an isomorphism on $\mathcal{R}_{H_l}^\diamond$, and
by cut point identification we known from Corollary~\ref{corr-cpi} that $\alpha_c(n)=\alpha_l(n)$ for all $n \in CP_{H_l}$.
\end{proof}

\section{Model checking Shylock programs}\label{sec-mc}

In this section we present a framework for model checking Shylock programs. We
first turn our abstract semantics into a pushdown system, then we introduce
a linear time temporal logic for heap structures, and finally we shortly recall
the actual model checking procedure.

\paragraph{Programs as pushdown systems}
A \emph{pushdown system} can be considered as a pushdown automaton without an input alphabet.
 Formally a pushdown system $\mathcal{P}$ is a triple $(\Delta, \Sigma, \mapsto)$ where
$\Delta$ is a set of \emph{control locations}, $\Sigma$ is a \emph{stack alphabet}, and
$\mapsto$ is a subset of $ \langle\Delta \times \Sigma\rangle \times \langle\Delta \times \Sigma^*\rangle$
representing the set of \emph{rules}. A pushdown system is said to be \emph{finite} when the
above three sets are all finite.

The behaviour of any Shylock program $P = \{p_0~::~B_0,\ldots,p_l~::~B_l\}$ can be represented by a pushdown system
$\mathcal{P}_P= (\Delta, \Sigma, \mapsto)$, where $\Delta$ is the set of all heaps, and
$\Sigma = \Delta \cup cl(P) \cup \{ Z \}$, where $Z$ is an element which does not occur in $\Delta$ and $cl(P)$. Here $cl(P)$ is
the set of all possibly reachable statements in $P$,
and it is defined as the union of all $cl(B_i)$, for all $ 0 \leq i \leq l$,  with $cl(B)$ given
inductively by:
\[
\begin{array}{lcl}
cl(x.f~{:}{=}~y) = \{x.f~{:}{=}~y\}        & \;\;\; &  cl(x~{:}{=}~y.f) = \{x~{:}{=}~y.f\}\\
cl(x~{:}{=}~\text{new}) = \{x~{:}{=}~\text{new}\} & & cl(p) = \{ p \}\\
cl([x=y] B) = \{[x=y] B\} \cup cl(B)  & & cl([x \neq y] B) = \{[x \neq y]  B\} \cup cl(B)\\
cl(B_1 + B_2) = cl(B_1) \cup cl(B_2) \cup \{B_1+ B_2\}  & &  cl(B_1;B_2) = cl(B_1) \cup cl(B_2)
\end{array}
\]
The rules of the pushdown system are specified using the \emph{abstract semantics} as follows:
\[
\langle H,\gamma \rangle \mapsto \langle H',w \rangle
~~~~~ \mbox{iff} ~~~~~ \langle H, \gamma\bullet\Gamma \rangle \longrightarrow \langle H', w\bullet\Gamma \rangle
\]
where $H$ ranges over heaps. Further we add rules $\langle H, \gamma \rangle \mapsto \langle H, Z\rangle$ for any configuration $\langle H, \gamma\rangle$ which does not have outgoing transitions to complete with stuttering steps terminating computations starting from $\langle H_0,p_0 \bullet Z \rangle$. Because there are infinitely many heaps, the
pushdown system constructed above will in general be infinite. Consequently
existing model checking techniques can not be applied.
In order to allow model checking we consider a subclass of programs. First, we need the following
definition:
\begin{definition}
 A heap $H$ is $k$-bounded if $|\mathcal{R}_H(G \cup L)| \leq k$. A computation
  $\langle H_0, \Gamma_0 \rangle \rightarrow \langle H_1, \Gamma_1 \rangle \rightarrow \ldots$
 (where the transition steps are according to the abstract semantics)
  is $k$-bounded if $|\mathcal{R}_{H_i}|$ is $k$-bounded for all $i$.
 A program $P$ with main procedure $p_0$ is $k$-bounded if every computation
  $\langle H, p_0 \rangle \rightarrow \ldots$ is $k$-bounded.
\end{definition}
\begin{example}
As an example of a program which is bounded in this sense, recall the program with
a single procedure defined as $p~::~x~{:}{=}~\text{new};~p$, where $x$ is a local variable. Indeed only one object
identity is needed to represent the object to which $x$ refers, so this program
is $1$-bounded. Nevertheless, since $x$ is local, during the execution of the program an unbounded number
of objects are stored on the stack.

For another example, recall the program from Example~\ref{ex-file} which opens and closes a file. This
program is $k$-bounded iff the visible heap during execution of the procedure $q(x)$ (with $x$ a fresh object) is $k$-bounded.
\qed
\end{example}

Now if we restrict to $k$-bounded programs, the abstract semantics,
because of its reuse of object identities, allows us to represent
the behaviour of a program as a pushdown system as above, but
using as control states only $k$-bounded heaps. By Theorem~\ref{thm-bisimulation}, we then have precise abstractions of
$k$-bounded programs as \emph{finite} pushdown systems. More precisely,
given a $k$-bounded program $P$ we define the pushdown system $k$-${\cal P}_P=(\Delta_k,\Sigma_k,\mapsto_k)$
obtained as a restriction from $\mathcal{P}_P$ as follows.
First, $\Delta_k = \{ H \mid | \mathcal{R}_H(G \cup L)|\leq k \}\cup\{ \top \}$ is the subset of all
$k$-bounded heaps. The stack alphabet $\Sigma_k$ is given by $\Delta_k \cup {\it cl}(P) \cup \{ Z \}$, and
the relation $\mapsto_k$ is the restriction of $\mapsto$ to $\Delta_k$ together with the
two out-of-bound rules below
\[
\frac{\langle H,\gamma \rangle \mapsto \langle H',w \rangle  \;\;\; | \mathcal{R}_{H'}| > k}
     {\langle H,\gamma\rangle \mapsto_k \langle \top ,\gamma\rangle}
\;\;\; \mbox{and} \;\;\;
\langle\top,\gamma\rangle \mapsto_k \langle\top,\gamma\rangle
\]
Note that in fact for any program $P$, $k$-$\mathcal{P}_P$ is a finite pushdown system. However it is a \emph{precise}
abstraction of $P$ only if $P$ is $k$-bounded.

\paragraph{Specifications in $L^iT_t{L}^e$}

In order to do a precise pointer analysis of Shylock programs, we introduce a
linear time temporal logic (${L}^i{T}_t{L}^e$) for describing the evolution
of the heap structure.
We first introduce the language of the properties satisfied by a given heap, which will
form the atomic properties of the linear temporal logic.
We do so by the introduction of expressions
of the Kleene algebra with tests~\cite{Kozen:1997:KAT:256167.256195} over fields and
variables. More precisely, let {\em Rite} be the smallest set defined by the following grammar:
\[
r~{:}{:}{=}~ \varepsilon \mid x\mid \lnot x \mid f\mid r.r\mid r+r\mid r^*
\]
where $x$ ranges over variable names (to be used as tests) and $f$ over field names
(to be used as actions). The regular expressions introduced by {\em Rite} are similar to the heap patterns used in matching logic~\cite{rosu-stefanescu-2011-tr} and separation logic~\cite{Reynolds2002}.
We define a transition relation $n \xrightarrow{r}_H m$
between objects of a heap $H$ as the least relation such that
\[
\begin{array}{ll}
n \xrightarrow{\varepsilon}_H n \\
n \xrightarrow{x}_H n        & \mbox{if $H(x)=n$}\\
n \xrightarrow{\neg x}_H n   & \mbox{if $H(x) \not= n$}\\
n \xrightarrow{f}_H m        & \mbox{if $H(f)(n)=m$}\\
n \xrightarrow{r_1 + r_2}_H m & \mbox{if $n \xrightarrow{r_1} m$ or $n \xrightarrow{r_2} m$} \\
n \xrightarrow{r_1.r_2}_H m & \mbox{if exists an object $n'$ such that $n \xrightarrow{r_1}_H n'$ and $n' \xrightarrow{r_2}_H m$}\\
n \xrightarrow{r^*}_H m     & \mbox{if either $n = m$ or there exists an object $n'$ such that $n \xrightarrow{r}_H n'$ and $n' \xrightarrow{r^*}_H m$}
\end{array}
\]
Further we introduce the following \emph{modal} interpretation of regular expressions:
$$
\mbox{\rm $H\models r$ if and only if for each reachable object $n\in \mathcal{R}_H$ there exists $m$ such that $n\xrightarrow{r}_H m$.}
$$
(Note that this coincides with the truth definition of $H \models \langle r \rangle {\it true}$ in dynamic logic.)
For instance, the regular expression
$
\it first.next^*.last + \neg first
$
is satisfied by a heap $H$ if and only if  the object referred to by the variable {\it first}
is linked via a chain of fields {\it next} with the object referred to  by the variable {\it last}, or
$\it first$ is not allocated.

Our ${L}^i{T}_t{L}^e$ formulas are built according to the following grammar:
$$\phi\;{:}{:}{=}\; true \mid r \mid \neg \phi \mid \phi_1 \wedge \phi_2 \mid {\rm X} \phi \mid \phi_1 {\rm U} \phi_2$$
where $r$ ranges over ${\it Rite}$. Other propositional connectives $\vee$, $\rightarrow$ are defined
in terms of $\wedge$ and $\neg$. Further we define ${\rm F} \phi = true \, {\rm U}\,\phi$ and ${\rm G} \phi= \neg {\rm F} (\neg \phi)$.
We define $At(\phi)$ to be
the set of atomic propositions $r \in {\it Rite}$ which appear in $\phi$. Clearly $At(\phi$) is always a
finite set.

For a set $X$ we denote its powerset by $2^X$.
With $(2^{At(\phi)})^{\omega} = \{w_0 w_1 w_2 \ldots | w_i \in 2^{At(\phi)} \text{ for all } i \geq 0\}$ we denote the set of
infinite words over sets of expressions in $At(\Phi)$. Given such an infinite word $w = w_0 w_1 w_2 \ldots$ we denote
with $w_i$ the $i$-th element, and with $w[i \ldots]$ the subsequence $w_i w_{i+1} \ldots$ starting from the $i$-th element of $w$.
For a ${L}^i{T}_t{L}^e$ formula $\phi$ and an infinite word $w \in (2^{At(\phi)})^{\omega}$
we denote that $w$ satisfies $\phi$ by $w \models \phi$. This satisfaction relation is defined inductively
on the structure of $\phi$ according to the standard semantics of LTL~\cite{BK08}:
$$
\begin{array}{ll}
 w \models true \\
 w \models r & \text{ iff } r \in w \\
 w \models \phi_1 \wedge \phi_2 & \text{ iff } w \models \phi_1 \text{ and } w \models \phi_2 \\
 w \models \neg \phi & \text{ iff } w \not \models \phi \\
 w \models X \phi & \text{ iff } w[1 \ldots] \models \phi \\
 w \models \phi_1 U \phi_2 & \text{ iff } \exists j \geq 0. w[j \ldots] \models \phi_2 \text{ and } w[i \ldots] \models \phi_1 \text{ for all } 0 \leq i < j
\end{array}
$$
Let $\pi$ be an infinite sequence of $k$-bounded heaps for some $k$, i.e., $\pi \in \{H_0  H_1 H_2 \ldots | H_i \in \Delta_k \text{ for all } i \geq 0\}$.
Intuitively, $\pi$ represents a trace of heaps which we encounter during a particular computation of a $k$-bounded program.
We say that $\pi \models \phi$ if and only if there exists a sequence $w \in (2^{At(\phi)})^{\omega}$ such that $w \models \phi$ and for all $i \geq 0$:
$$\pi_i \models r \text{ for all } r \in w_i$$
Finally the above relation $\models$ is pointwise extended to sets of infinite sequences of $k$-bounded heaps.



\paragraph{Model checking}Recall that a B\"uchi automaton $\mathcal{B} = (\nabla, A, \leadsto, Q_0, F)$ consists of a
finite set of states $\nabla$, an input alphabet $A$, a transition relation
$\leadsto \subseteq \nabla \times A \to \nabla$, a set of initial states $Q_0$ and a
set of final states $F$. The language $\mathcal{L}(\mathcal{B})$ accepted by
$\mathcal{B}$ is the set of all infinite words $w$ over $A$ such that there is an infinite
path via $\leadsto$ labeled by $w$, starting from a state $q_0 \in Q_0$, and visiting
an accepting state in $F$ infinitely often. Given a ${L}^i{T}_t{L}^e$ formula $\phi$,
one can effectively construct a B\"uchi automaton $\mathcal{B}_\phi$ which recognizes exactly
the set of words $w$ over sets of expressions in $At(\phi)$ satisfying $\phi$ (see e.g.~\cite{BK08} for details).

Let $\phi$ be a formula of our temporal logic for heaps, and $P$ be a
$k$-bounded Shylock program. To check if $P$
satisfies the formula $\phi$, we have to check if all the computations
of the pushdown system $k$-${\cal P}_P$ starting from the initial configuration
$\langle H_0,p_0\bullet Z\rangle$ satisfy $\phi$. This amounts to synchronizing the pushdown
system $k$-$\mathcal{P}_P$ with the B\"uchi automaton $\mathcal{B}_{\neg \phi}$ and
checking if the resulting B\"uchi pushdown system has an accepting run, i.e., a run
starting from the initial states of the two systems which visits infinitely
often configurations whose control locations projected into the states of the
B\"uchi automaton are final~\cite{ES01,Sch02b}.

The problem of finding an accepting run of a B\"uchi pushdown system can be
reduced to that of finding a \emph{repeated head} reachable from the initial
configuration~\cite{ES01,Sch02b}. Computing the repeating heads is typically
developed in two phases. In the first phase, one constructs the \emph{head reachability
graph} $\mathcal{G}$ associated with a  B\"uchi pushdown system, while in the second
phase, $\mathcal{G}$ is analyzed to identify those nodes of the graph which are
repeated heads. To avoid redundant computations, it suffices to construct the
\emph{head reachability graph} $\mathcal{G}$ restricted to those configurations
reachable from an initial configuration of the B\"uchi pushdown system. This can
be done using forward reachability analysis, by using the so called ${\it post}^*$
method. For more details, see, e.g., Chapter 3 in~\cite{Sch02b}.


\paragraph{Shylock and \cite{BFQ07}} As discussed in the introduction, the related work closest to our approach is presented in \cite{BFQ07}. While \cite{BFQ07} considers only reachability, this could be extended to a full model checking procedure along the lines discussed in this paper.
A first view on the two works shows that \cite{BFQ07} has a very high level solution to the problem while we give an explicit solution to it.  The semantics of procedure call and specifically of the procedure return as given in \cite{BFQ07} are stated in terms of abstract graph isomorphisms, and it is not clear how they should be implemented. In contrast, in this paper we give a purely symbolic characterization of procedure call and return.

Now we present a more detailed view on the difference between the semantics given in the two works. The procedure call in  \cite{BFQ07} employs  the cut point mechanism as well, but it also \emph{cleans} the heap of the currently unreachable objects. Hence, their procedure call passes to the callee only the strictly visible heap of the caller, actually an isomorphic instance of it. In contrast, Shylock's procedure call relies on the cut points as well, but it doesn't necessarily clean the heap (though it could). Instead, Shylock \emph{reuses} object identities on demand, during the object creation statements of the procedure, while  \cite{BFQ07} doesn't pay attention to the body of the procedure because of the initial cleaning. Though Shylock may seem lazy w.r.t. cleaning, it's memory reuse mechanism acts in fact as a \emph{localized} cleaning. This  pays off during the execution of the procedure returns. Namely, in  \cite{BFQ07} the procedure return has to proceed by  renaming the \emph{entire} visible object space, such that it can synchronize the current heap with the caller's heap. Meanwhile, Shylock renames \emph{only} the name clashes, i.e., the objects at the intersection of the caller's global heap and the callee's  current purely local heap. These differences induce a different reasoning during the verification phase. Namely, while Shylock can afford to use heap equality during the model checking phase, the reachability procedure in  \cite{BFQ07} has to be performed on normal forms of the heaps (i.e., on the representatives of the graph isomorphic equivalence class).
We are not sure if Shylock maintains strictly one representative of each isomorphic class, but we plan to study this particular aspect in the near future.

\section{Conclusions}\label{sec-concl}

In the presence of recursive procedures and local variables, an unbounded number of objects can be allocated
either on the call stack using local variables, or, anonymously, on the heap using reference fields.
In this paper we discussed \emph{Shylock}, a language which supports these features,
together with a formal abstract semantics which allows model checking in
the context of bounded visible heaps. We introduced
a temporal logic for specifying properties of the heap, and
discussed a procedure for checking these properties against Shylock programs.

\paragraph{Future work}
Shylock's improved semantics has been implemented in the $\mathbb{K}$ framework. We are currently
implementing a general model checking technique for recursive programs defined in $\mathbb{K}$, from
which we would obtain a Shylock model checker along the lines described in this paper. Further
we are investigating the expressive power of programs with a bounded visible heap.

\nocite{*}
\bibliographystyle{eptcs}

\end{document}